\documentclass[runningheads,11pt]{llncs}

\usepackage{geometry}
\usepackage[utf8]{inputenc}
\usepackage[T1]{fontenc}
\usepackage{lmodern}
\usepackage{amsfonts}
\usepackage{amsmath,mathtools}
\usepackage{amssymb}
\usepackage{hyperref}
\hypersetup{
	colorlinks,
	citecolor=black,
	filecolor=black,
	linkcolor=black,
	urlcolor=black
}
\usepackage{algorithm2e}
\usepackage{mathalfa}
\numberwithin{equation}{section}
\usepackage{bm}
\usepackage{spverbatim}
\usepackage{breakcites}
\usepackage{enumitem}
\renewcommand{\qed}{\hfill \blacksquare}
\usepackage{relsize}
\usepackage{nccmath}
\usepackage{cleveref}
\usepackage{enumitem}
\usepackage{multicol,xparse}
\usepackage{changepage}
\usepackage{mdframed}
\usepackage[sanserif,basic]{complexity}

\usepackage{upgreek}

\makeatletter
\def\@fnsymbol#1{\ensuremath{\ifcase#1\or \dagger\or \ddagger\or
		\mathsection\or \mathparagraph\or \|\or **\or \dagger\dagger
		\or \ddagger\ddagger \else\@ctrerr\fi}}
\makeatother

\DeclareMathOperator{\Char}{char}

\title{Function-private Conditional Disclosure of Secrets and Multi-evaluation Threshold Distributed Point Functions\thanks{The is the full version of the paper that will appear in Cryptology and Network Security (CANS), 2021.}}
\author{Nolan Miranda \inst{1,3}, Foo Yee Yeo \inst{2}, Vipin Singh Sehrawat \inst{1}}
\institute{Seagate Technology, USA \and Seagate Technology, Singapore \and Stanford University, USA}

\begin{document}
	\maketitle 
\begin{abstract} \normalsize
	Conditional disclosure of secrets (CDS) allows multiple parties to reveal a secret to a third party if and only if some pre-decided condition is satisfied. In this work, we bolster the privacy guarantees of CDS by introducing function-private CDS wherein the pre-decided condition is never revealed to the third party. We also derive a function secret sharing scheme from our function-private CDS solution. The second problem that we consider concerns threshold distributed point functions, which allow one to split a point function such that at least a threshold number of shares are required to evaluate it at any given input. We consider a setting wherein a point function is split among a set of parties such that multiple evaluations do not leak non-negligible information about it. Finally, we present a provably optimal procedure to perform threshold function secret sharing of any polynomial in a finite field. 
\end{abstract}	

\section{Introduction}
In 1994, De Santis et al.~\cite{Santis[94]} introduced the concept of function secret sharing (FSS) as a special case of secret sharing~\cite{Shamir[79],Blakley[79]} wherein a class $\mathcal{F}$ of efficiently computable and succinctly described functions $f_i: \{0,1\}^\ell \longrightarrow\mathbb{G}$, where $\mathbb{G}$ is a group, is shared among a set of parties. FSS allows a dealer to randomly split an arbitrary function $f \in \mathcal{F}$ into $n \geq 2$ shares $\{f_i\}_{i=1}^n$ such that only authorized subsets of shares can be combined to evaluate or to reconstruct $f$. In 2014, Gilboa and Ishai~\cite{Niv[14]} considered a specialization of FSS called distributed point functions (DPF), which is an FSS scheme for a family of point functions $\mathcal{P}(A,B)=\{p_{a,b}: a\in A,\,b\in B\}$, whose members are defined as: 
$$p_{a,b}(x) = 
\begin{cases} 
	b & x = a,\\
	0 & \text{otherwise}.
\end{cases}
$$ 

A DPF scheme is said to be threshold if for some fixed $t \leq n$, all subsets of function shares $\{p_i\}_{i=1}^n$ with cardinality at least $t$ can be used to reconstruct $p_{a,b}$. DPF has found applications in privacy-preserving primitives such as private information retrieval~\cite{Niv[14]} and private contact tracing~\cite{Ditti[20]}.

The second topic whose existing solutions we improve upon is called conditional disclosure of secrets (CDS)~\cite{Gertner[98]}, which allows two or more parties, holding some input, to share a secret $s$ with an external party Carol such that $s$ is revealed only if some fixed condition holds on the joint input and shared randomness of the parties. The condition is often encoded as some Boolean condition function $h: \mathcal{C} \to \{0,1\}.$ The parties are allowed to send a single message to Carol which may depend on their ``shares'' and shared randomness. CDS has found multiple applications in cryptography, including private information retrieval~\cite{Gertner[98]}, attribute-based encryption~\cite{RomaIor[15]}, priced oblivious transfer~\cite{AieIsh[01]} and secret sharing for  uniform/general/forbidden/graph access structures~\cite{BFMP[20],BFMP[17],Apple[19],Liu[18],Benny[20]}. 

\subsection{Our Contributions}
Our contributions are multifold, which can be described as:
\subsubsection{Function-private CDS.}
Content-based filtering is a tried and tested subroutine used in generating content recommendations. It generates the intended information for users by comparing representations of information search to those of the contents extracted from user profiles, indicating users' interests. Recently, streaming services have been establishing collaborations wherein they provide services and content recommendations based on their combined data. For instance, Spotify, Hulu and Showtime announced such a collaboration in 2018~\cite{Spotify[18]}. Moreover, almost all major streaming services have migrated to the cloud~\cite{GoogleSpotify[1],AmazonNetflix[1],AmazonHulu[1]}. We know from~\cite{Anick[89],JoonHo[94],VerGoff[61]} that it is possible for these streaming services to encode their content recommendations generating content-based filtering as Boolean functions, which reveal the recommendations for various demographics only when the function evaluates to ``true''. Since multiple collaborating service providers use the same cloud for their data and recommendations' generation, the third-party cloud service provider serves as Carol in these settings. Hence, CDS fits well in such settings. However, with their services outsourced to the cloud, the service provides may want to hide information about their recommendations suggesting algorithms. Hence, it is desirable to have a CDS scheme that hides the condition function(s) from Carol, even after the secret, i.e., recommendations for the target demographic, are revealed.   

We address this requirement by introducing function-private CDS, which we define as follows: let $h$ belong to a family $\mathcal{H}: \mathcal{C} \to \{0,1\}$ of Boolean condition functions, then a function-private CDS scheme satisfies the following conditions:
\begin{enumerate}
	\item Correctness: for every $c \in \mathcal{C}$, if it holds that $h(c) = 1$, then the correct secret $s$ is revealed to Carol with probability 1. Else, if $h(c) = 0$, then Carol rejects with probability 1.
	\item Secrecy: for any $c \in \mathcal{C}$ such that $h(c) = 0$, it holds that Carol cannot gain any information about the secret. 
	\item Function privacy: for certain conditions on the parties' inputs, it holds that no party has any non-negligible advantage in distinguishing the Boolean condition function $h \in \mathcal{H}$ from every $h' \in \mathcal{H}$, where $h'(c) = h(c)$. 
	\item Input privacy: for any $c, c' \in \mathcal{C}$ such that $h(c) = h(c')$, it holds that Carol has no non-negligible advantage in distinguishing between $c$ and $c'$.
\end{enumerate}
Function privacy has been studied in the context of functional encryption~\cite{ShaAg[13],DanAna[13],DanAnaGil[13],ZviGil[15],FanTang[18],AbdaDavid[21]} but ours is the first scheme to achieve it for CDS. For an introduction to functional encryption, we refer the interested reader to~\cite{Amali[21]}.
\subsubsection{Function-private CDS to FSS.}
Given a function-private CDS scheme for $k$ parties $P_1, \ldots, P_k$ (excluding Carol) for the family of Boolean condition functions $\mathcal{H}$, secret domain $\mathcal{S}$, and randomness $r$, we prove that there exists a $k$-out-of-$k$ function secret sharing scheme for the family $\mathcal{H}$. To the best of our knowledge, this is the first derivation of FSS from CDS.
\subsubsection{Multi-evaluation Threshold DPF.}
In an example in \cite{BBDK[00]}, Beimel et al. introduced threshold function secret sharing of a family of point functions $\mathcal{P}(A,B)$. We identify a weakness in that example. Specifically, their example approach leaks information about the point function $p_{a,b} \in \mathcal{P}(A,B)$ if multiple evaluations are performed. We rectify this issue by extending their solution such that repeated evaluations do not leak non-negligible information about $p_{a,b}$. Our scheme, called multi-evaluation threshold DPF, uses a key-homomorphic pseudorandom function (PRF) family. \\[1.5mm]
\textit{Key-homomorphic PRF:} In a PRF family~\cite{Gold[86]}, each function is specified by a key such that it can be evaluated deterministically given the key whereas it behaves like a random function without the key. For a PRF $F_k$, the index $k$ is called its key/seed. A PRF family $F$ is called key-homomorphic if the set of keys has a group structure and there is an efficient algorithm that, given $F_{k_1}(x)$ and $F_{k_2}(x)$, outputs $F_{k_1 \oplus k_2}(x)$, where $\oplus$ is the group operation~\cite{Naor[99]}.
\subsubsection{Optimal Threshold FSS for Polynomials.} 
We present a novel method to perform threshold FSS for any polynomial in a finite field. We prove that the share size of our scheme is optimal for FSS over polynomials.  

\subsection{Organization}
The remaining text is organized as follows: \Cref{sec2} recalls the concepts required for the rest of the paper. In \Cref{sec3}, we present our multi-evaluation threshold DPF scheme. In \Cref{sec4}, we define function-private CDS and introduce the first scheme to realize it. In the same section, we provide a mechanism to extend any function-private CDS scheme into an FSS scheme. \Cref{sec5} details our optimal scheme for FSS over polynomials. In \Cref{sec6}, we discuss interesting problems for future work in FSS. 

\section{Preliminaries}\label{sec2}

\begin{definition}[Negligible Function]\label{Neg}
	\emph{For security parameter $\lambda$, a function $\epsilon(\lambda)$ is called \textit{negligible} if for all $c > 0$, there exists a $\lambda_0$ such that $\epsilon(\lambda) < 1/\lambda^c$ for all $\lambda > \lambda_0$.}
\end{definition}

\begin{definition}[Computational Indistinguishability~\cite{Gold[82]}]
	\emph{Let $X = \{X_\lambda\}_{\lambda \in \mathbb{N}}$ and $Y = \{Y_\lambda\}_{\lambda \in \mathbb{N}}$ be ensembles, where $X_\lambda$'s and $Y_\lambda$'s are probability distributions over $\{0,1\}^{\kappa(\lambda)}$ for $\lambda \in \mathbb{N}$ and some polynomial $\kappa(\lambda)$. We say that $\{X_\lambda\}_{\lambda \in \mathbb{N}}$ and $\{Y_\lambda\}_{\lambda \in \mathbb{N}}$ are polynomially/computationally indistinguishable if the following holds for every (probabilistic) polynomial-time algorithm $\mathcal{D}$ and all $\lambda \in \mathbb{N}$:
	\[\Big| \Pr[t \leftarrow X_\lambda: \mathcal{D}(t) = 1] - \Pr[t \leftarrow Y_\lambda: \mathcal{D}(t) = 1] \Big| \leq \epsilon(\lambda),\]
	where $\epsilon$ is a negligible function.}
\end{definition}	

\begin{remark}[Perfect Indistinguishability]
	We say that $\{X_\lambda\}_{\lambda \in \mathbb{N}}$ and $\{Y_\lambda\}_{\lambda \in \mathbb{N}}$ are perfectly indistinguishable if the following holds for all $t$:
	\[\Pr[t \leftarrow X_\lambda] = \Pr[t \leftarrow Y_\lambda].\]
\end{remark}

Consider adversaries interacting as part of probabilistic games. For an adversary $\mathcal{A}$ and two games $\mathfrak G_0, \mathfrak G_1$ with which it can interact, $\mathcal{A}'s$ distinguishing advantage is: $Adv_{\mathfrak G_0, \mathfrak G_1}(\mathcal{A}) := \Big|\Pr[\mathcal{A} \text{ accepts in } \mathfrak G_0] - \Pr[\mathcal{A} \text{ accepts in } \mathfrak G_1]\Big|.$ For security parameter $\lambda$ and a negligible function $\epsilon$, the two games are said to be computationally indistinguishable if it holds that: $Adv_{\mathfrak G_0,\mathfrak G_1}(\mathcal{A}) \leq \epsilon(\lambda)$.

\begin{definition}[PRF] 
	\emph{Let $A$ and $B$ be finite sets, and let $\mathcal{F} = \{ F_k: A \rightarrow B \}$ be a function family, endowed with an efficiently sampleable distribution (more precisely, $\mathcal{F}$, $A$ and $B$ are all indexed by the security parameter $\lambda)$. We say that $\mathcal{F}$ is a PRF family if the following two games are computationally indistinguishable:
	\begin{enumerate}
		\item Choose $F_k \in \mathcal{F}$ and give the adversary adaptive oracle access to $F_k(\cdot)$.
		\item Choose a uniformly random function $U: A \rightarrow B$ and give the adversary adaptive oracle access to $U(\cdot).$
	\end{enumerate}}
\end{definition} 

Numerous PRF families with various useful properties have been constructed~\cite{Naor[99],Boneh[13],Ban[14],Parra[16],SamK[20],Vipin[19],Zvika[15],KimDan[17],RotBra[17],RanChen[17],KimWu[17],KimWu[19],Qua[18]}. For a detailed introduction to PRFs and review of the noteworthy results, we refer the interested reader to \cite{AndAlo[17]}.

\section{Multi-evaluation Threshold DPF}\label{sec3}
In this section, we introduce multi-evaluation threshold DPFs, and present a scheme to realize it.  

\begin{definition}
	\emph{Given a string $a \in \{0,1\}^\ell$ and a value $\alpha \in \mathbb{F}$, a computational multi-evaluation distributed point function scheme for a $t$-out-of-$n$ threshold structure is defined as a collection of three algorithms $(\mathsf{Gen}, \mathsf{Eval}, \mathsf{Rec})$ such that:}
	\begin{itemize} 
		\item \emph{A randomized algorithm $\mathsf{Gen}$ takes three inputs, $a\in\{0,1\}^\ell$, $\alpha\in\mathbb{F}$ and a security parameter $\lambda\in\mathbb{Z}^+$, and generates $n$ keys $\{f_i\}_{i=1}^n$, representing secret shares of a dimension-$2^\ell$ vector $\mathbf{v}$ that has value $\alpha \in \mathbb{F}$ only at the $a$-th position and is zero at every other position.
			\item A deterministic algorithm $\mathsf{Eval}$ that takes three inputs, a key $f_i~(i \in [n])$, $x \in \{0,1\}^\ell$ and some $r\in R$, and outputs a share $s_i$.
			\item A deterministic algorithm $\mathsf{Rec}$ that takes the outputs of $\mathsf{Eval}$ from $t$ parties and outputs an element of $\mathbb{F}$.}
	\end{itemize}
	\emph{These algorithms satisfy the following three conditions:}
	\begin{itemize}
		\item Computational Correctness: \emph{for all strings $a \in \{0,1\}^\ell$, output values $\alpha \in \mathbb{F}$, $\lambda\in\mathbb{Z}^+$, $r\in R$, keys $\{f_i\}_{i=1}^n \leftarrow \mathsf{Gen}(a, \alpha, \lambda)$ and subsets $T\subseteq [n]$ of size $t$, it holds that: $\Pr\left[\mathsf{Rec}\left(\{\mathsf{Eval}(f_i, a, r)\}_{i\in T}\right) = \alpha\right] = 1$, and for all strings $x \in \{0,1\}^l$ such that $x\neq a$, $\Pr\left[\mathsf{Rec}\left(\{\mathsf{Eval}(f_i, x, r)\}_{i\in T}\right) = 0\right] > 1-\text{negl}(\lambda).$
		}
		\item Perfect Secrecy: \emph{for all strings $a, b \in \{0,1\}^\ell$, output values $\alpha, \beta \in \mathbb{F}$, $\lambda\in\mathbb{Z}^+$, keys $\{f_i\}_{i=1}^n \leftarrow \mathsf{Gen}(a, \alpha, \lambda)$ and $\{f_i'\}_{i=1}^n \leftarrow \mathsf{Gen}(b, \beta, \lambda)$ and subset $S\subset [n]$ of size $<t$, it holds that $\{f_i\}_{i\in S}$ and $\{f_i'\}_{i\in S}$ are perfectly indistinguishable.}
		\item Computational Multi-evaluation: \emph{for all strings $a, b \in \{0,1\}^\ell$, output values $\alpha, \beta \in \mathbb{F}$, $\lambda\in\mathbb{Z}^+$, keys $\{f_i\}_{i=1}^n \leftarrow \mathsf{Gen}(a, \alpha, \lambda)$ and $\{f_i'\}_{i=1}^n \leftarrow \mathsf{Gen}(b, \beta, \lambda)$, it holds for all strings $x_1,\,x_2,\,\ldots,\,x_m \neq a, b$ and $r_1,\,r_2,\,\ldots,\,r_m\in R$ distinct, and subset $S\subset [n]$ of size $<t$, that:
			$$\left(\{f_i\}_{i\in S},\ \{\mathsf{Eval}(f_i, x_h, r_h)\}_{i\in [n], h\in [m]}\right), \left(\{f'_i\}_{i\in S},\ \{\mathsf{Eval}(f_i', x_h, r_h)\}_{i\in [n], h\in [m]}\right)$$
			are computationally indistinguishable w.r.t $\lambda$.
		}
	\end{itemize}
\end{definition}

\subsection{$n$-out-of-$n$ Multi-evaluation DPF}
Here, we describe an $n$-out-of-$n$ computational multi-evaluation DPF for the class of point functions $\mathcal{P}\left(\{0,1\}^\ell,\mathbb{F}\right)$, where $\mathbb{F}=\mathbb{F}_q$ is the finite field with cardinality $q$. Let $\mathcal{F}=\{F^{(\lambda)}: \mathcal{K}^{(\lambda)}\times R\to\mathbb{F}^{2\ell+\lambda+1}\}$
be a family of key-homomorphic PRFs such that the advantage of any polynomial-time adversary in distinguishing $F^{(\lambda)}$ from random is negligible in $\lambda$, and such that $F^{(\lambda)}(k_1+k_2,r)=F^{(\lambda)}(k_1,r)+F^{(\lambda)}(k_2,r)$ for all $k_1,k_2\in\mathcal{K}^{(\lambda)}$ and $r\in R$. Write $F^{(\lambda)}=(F^{(\lambda)}_1, F^{(\lambda)}_2)$ with $F^{(\lambda)}_1: \mathcal{K}^{(\lambda)}\times R\to\mathbb{F}^{2\ell+\lambda}$ and $F^{(\lambda)}_2:\mathcal{K}^{(\lambda)}\times R\to\mathbb{F}$. For conciseness, we write $F$ for $F^{(\lambda)}$ (and $F_k$ for $F^{(\lambda)}_k$, $\mathcal{K}$ for $\mathcal{K}^{(\lambda)}$) when $\lambda$ is clear from context.

We make the assumption that $\mathcal{K}^{(\lambda)}$ is an abelian group and that the order of any element in $\mathcal{K}^{(\lambda)}$ is bounded by some polynomial $\gamma(\lambda)$. (This is often the case; in particular, this holds when $\mathcal{K}^{(\lambda)}=\left(\mathbb{Z}/h(\lambda)\mathbb{Z}\right)^{g(\lambda)}$, where $g(\lambda)$ is an arbitrary function of $\lambda$ and $h(\lambda)$ is polynomially bounded.) Since $|\mathcal{K}^{(\lambda)}|$ is superpolynomial in $\lambda$, if the above conditions hold, then there exists $\lambda_0$ such that $\frac{\gamma(\lambda)^{2\ell n}}{|\mathcal{K}^{(\lambda)}|}<1-\frac{1}{\lambda}$ for all $\lambda\geq\lambda_0$. Hence, by replacing $\lambda$ by a larger $\lambda'$ if needed and truncating the output, we may assume that $\frac{\gamma(\lambda)^{2\ell n}}{|\mathcal{K}^{(\lambda)}|}<1-\frac{1}{\lambda}$ holds for all $\lambda$.

\begin{remark}
	\label{KH-PRF_remark}
	If $\mathbb{F}$ has characteristic $p$, then for any $k\in p\mathcal{K}$, $r\in R$,
	$$F(k, r)=F(pk',r)=pF(k',r)=0.$$
	Thus, any key $k\in p\mathcal{K}$ is a ``weak key'', and since $F$ is a secure PRF, $|\mathcal{K}/p\mathcal{K}|^{-1}$ must be a negligible function of $\lambda$. By the fundamental theorem of finite abelian groups, we can write
	$$\mathcal{K}\cong\mathbb{Z}/(p_1^{n_1}\mathbb{Z})\times\mathbb{Z}/(p_2^{n_2}\mathbb{Z})\times\cdots\times\mathbb{Z}/(p_l^{n_l}\mathbb{Z})$$
	where $p_1,\ldots,p_l$ are (not necessarily distinct) primes. Assume $p_i=p$ for $1\leq i\leq l'$ and that $p_i\neq p$ for $l'<i\leq l$. Then $\mathcal{K}/p\mathcal{K}\cong (Z/p\mathbb{Z})^{l'}$ and thus $|\mathcal{K}/p\mathcal{K}|^{-1}=1/p^{l'}$ must be a negligible function of $\lambda$.
\end{remark}

Our scheme is a collection of three algorithms, $(\mathsf{Gen}(a,\alpha,\lambda)$, $\mathsf{Eval}(f_i, x, r)$, $\mathsf{Rec}(s_1,\ldots,s_n))$, which are defined as:

$$\bm{\mathsf{Gen}(a,\alpha,\lambda)}$$
\begin{enumerate}[topsep=-3pt]
	\item Choose $2\ell$ random vectors $v_0, v_1, v_2, \ldots, v_{2\ell-1}$ from $\mathbb{F}^{2\ell+\lambda}$.
	\item Choose $2\ell n$ random vectors $v_{i,j}\in\mathbb{F}^{2\ell+\lambda}$ ($1\leq i\leq n$, $0\leq j\leq 2\ell-1$) subject to the condition $v_j=\sum_{i=0}^n v_{i,j}$ for all $j$.
	\item Let $a = a_1a_2\ldots a_\ell$ and compute $\theta=\sum_{j=0}^{\ell-1} v_{2j+a_j}$. This sum includes either $v_{2j}$ or $v_{2j+1}$ depending on whether the $j$-th bit of $a$ is 0 or 1 respectively.
	\item Choose $2\ell$ random elements $\alpha_0, \alpha_1, \ldots, \alpha_{2\ell-1}\in\mathbb{F}$ subject to the condition $\alpha=\sum_{j=0}^{\ell-1}\alpha_{2j+a_j}$.
	\item Choose $2\ell n$ random elements $\alpha_{i,j}\in\mathbb{F}$ ($1\leq i\leq n$, $0\leq j\leq 2\ell-1$) subject to the condition that $\alpha_j=\sum_{i=1}^n\alpha_{i,j}$ for all $j$.
	\item Choose $2\ell n$ linearly independent keys $k_{i,j}$ ($1\leq i\leq n$, $0\leq j\leq 2\ell-1$).
	\item Compute $k=\sum_{i=1}^n \sum_{j=0}^{\ell-1} k_{i,2j+a_j}$.
	\item Output $f_i=(v_{i,0}, \ldots, v_{i,2\ell-1}, \theta, \alpha_{i,0}, \ldots, \alpha_{i,2\ell-1}, k_{i,0}, \ldots, k_{i,2\ell-1}, k)$.
\end{enumerate}

$$\bm{\mathsf{Eval}(f_i, x, r)}$$
\begin{enumerate}[topsep=-3pt]
	\item Parse $f_i$ as $(v_{i,0}, \ldots, v_{i,2\ell-1}, \theta, \alpha_{i,0}, \ldots, \alpha_{i,2\ell-1}, k_{i,0}, \ldots, k_{i,2\ell-1}, k).$
	\item Let $x = x_1x_2\ldots x_\ell$. Compute $s_{i,0}=\sum_{j=0}^{\ell-1}\left(v_{i,2j+x_j} + F_1(k_{i,2j+x_j},r)\right)$.
	\item Compute $s_{i,1}=\sum_{j=0}^{\ell-1}\left(\alpha_{i,2j+x_j} + F_2(k_{i,2j+x_j},r)\right)$.
	\item Output $s_i=(s_{i,0}, s_{i,1}, r, \theta, k)$
\end{enumerate}

$$\bm{\mathsf{Rec}(s_1,\ldots,s_n)}$$
\begin{enumerate}[topsep=-3pt]
	\item Parse $s_i$ as $(s_{i,0}, s_{i,1}, r, \theta, k)$.
	\item Compute $\sum_{i=0}^n s_{i,0}$. If this equals $\theta+F_1(k,r)$, output $\sum_{i=0}^n s_{i,1}-F_2(k,r)$ else output $0$.
\end{enumerate}

\begin{remark}
	In the above scheme, each party has a share size of
	$$(4\ell^2+2\lambda\ell+4\ell+\lambda)\log|\mathbb{F}|+(2\ell+1)\log|\mathcal{K}|,$$
	and the output of $\mathsf{Eval}$ for each party has size
	$$(4\ell+2\lambda+1)\log|\mathbb{F}|+\log|\mathcal{K}|+\log|R|,$$
	both of which are independent of the number of parties.
\end{remark}

\begin{theorem}
	\label{n_out_of_n_DPF}
	The above scheme is an $n$-out-of-$n$ computational multi-evaluation DPF scheme for sharing the class of point functions $\mathcal{P}\left(\{0,1\}^\ell,\mathbb{F}\right)$.
\end{theorem}

\begin{proof}
	\emph{Computational Correctness}: We first prove that evaluation at $x=a$ gives the correct result with probability $1$, i.e., $\mathsf{Rec}\left(\{\mathsf{Eval}(f_i, a, r)\}_{i\in [n]}\right) = \alpha$.
	Note that:
	\begin{align*}
		\textstyle{\sum}_{i=0}^n s_{i,0}&=\textstyle{\sum}_{i=0}^n\textstyle{\sum}_{j=0}^{\ell-1}\left(v_{i,2j+a_j}+F_1(k_{i,2j+a_j},r)\right) \\
		&=\textstyle{\sum}_{j=0}^{\ell-1}\textstyle{\sum}_{i=0}^n v_{i,2j+a_j}+\textstyle{\sum}_{i=0}^n\textstyle{\sum}_{j=0}^{\ell-1}F_1\left(k_{i,2j+a_j},r\right) \\
		&=\textstyle{\sum}_{j=0}^{\ell-1} v_{2j+a_j}+F_1\left(\textstyle{\sum}_{i=0}^n\textstyle{\sum}_{j=0}^{\ell-1}k_{i,2j+a_j},r\right) =\theta+F_1(k, r),
	\end{align*}
	hence, the output of $\mathsf{Rec}$ is:
	\begin{align*}
		\textstyle{\sum}_{i=0}^n s_{i,1}-F_2(k,r)	=\ &\textstyle{\sum}_{i=0}^n\textstyle{\sum}_{j=0}^{\ell-1}\left(\alpha_{i,2j+a_j}+F_2(k_{i,2j+a_j},r)\right)-F_2(k,r) \\
		=\ &\textstyle{\sum}_{j=0}^{\ell-1}\textstyle{\sum}_{i=0}^n\alpha_{i,2j+a_j}+\textstyle{\sum}_{i=0}^n\textstyle{\sum}_{j=0}^{\ell-1}F_2(k_{i,2j+a_j},r)-F_2(k,r) \\
		=\ &\textstyle{\sum}_{j=0}^{\ell-1}\alpha_{2j+a_j}+F_2\left(\textstyle{\sum}_{i=0}^n\textstyle{\sum}_{j=0}^{\ell-1}k_{i,2j+a_j},r\right)-F_2(k,r) \\
		=\ &\alpha+F_2(k,r)-F_2(k,r) = \alpha.
	\end{align*}
	Next, we prove that evaluation at $x\neq a$ is correct except with probability negligible in $\lambda$. Let $u_j=v_j+F_1(k_j,r)$ for $j=0,\ldots, 2\ell-1$, where $k_j=\sum_{i=1}^n k_{i, j}$. A simple calculation shows that:
	$$\textstyle{\sum}_{j=0}^{\ell}u_{2j+a_j}=\textstyle{\sum}_{j=0}^{\ell-1} \left(v_{2j+a_j}+F_1(k_{2j+a_j},r)\right)=\theta+F_1(k, r).$$
	Since $F_1$ is a PRF and $k_0,k_1,\ldots,k_{2\ell-1}$ are linearly independent, the vectors $u_0,u_1,\ldots,u_{2\ell-1}$ cannot be distinguished from random vectors in $\mathbb{F}^{2\ell+\lambda}$ except with probability negligible in $\lambda$. The probability that $2\ell$ random vectors are linearly independent in $\mathbb{F}^{2\ell+\lambda}$ is:
	\begin{align*}
		\textstyle{\prod}_{j=0}^{2\ell-1}\left(\frac{q^{2l+\lambda}-q^j}{q^{2\ell+\lambda}}\right)&=\textstyle{\prod}_{j=0}^{2\ell-1}\left(1-\frac{q^j}{q^{2\ell+\lambda}}\right)>1-\sum_{j=0}^{2\ell-1}\left(\frac{q^j}{q^{2\ell+\lambda}}\right) \\
		&=1-\textstyle{\frac{1}{q^{2\ell+\lambda}}}\left(\textstyle{\frac{q^{2\ell}-1}{q-1}}\right)>1-\textstyle{\frac{1}{q^\lambda}}=1-\text{negl}(\lambda).
	\end{align*}
	If the vectors $u_0,u_1,\ldots,u_{2\ell-1}$ are linearly independent, then there is no other linear combination of the $u_j$'s that result in $\theta+F_1(k,r)$, and thus, $\mathsf{Rec}$ outputs $0$ when given as inputs the outputs of $\mathsf{Eval}$ evaluated at $x\neq a$. Therefore, the output of $\mathsf{Rec}$ is $0$ except with probability negligible in $\lambda$.
	
	\emph{Perfect Secrecy}: Recall that $\mathsf{Gen}(a,\alpha,\lambda)$ outputs $(f_1,f_2,\ldots,f_n)$, where:
	$$f_i=(v_{i,0}, \ldots, v_{i,2\ell-1}, \theta, \alpha_{i,0}, \ldots, \alpha_{i,2\ell-1}, k_{i,0}, \ldots, k_{i,2\ell-1}, k).$$
	For $f_i$'s supplied by $n-1$ parties, which we assume, without loss of generality, to be the first $n-1$ parties, note that $v_{i,j}$ ($1\leq i\leq n-1$, $0\leq j\leq 2\ell-1$) and $\theta$ are independent elements (in the probabilistic sense) from the uniform distribution on $\mathbb{F}^{2\ell+\lambda}$, $\alpha_{i,j}$ ($1\leq i\leq n-1$, $0\leq j\leq 2\ell-1$) are independent elements from the uniform distribution on $\mathbb{F}$, while $k_{i,j}$ ($1\leq i\leq n-1$, $0\leq j\leq 2\ell-1$) and $k$ are $2\ell(n-1)+1$ linearly independent elements picked uniformly at random from $\mathcal{K}$. Thus, $(f_1,\ldots,f_{n-1})$ has the same distribution regardless of the value of $a\in\{0,1\}^\ell$ and $\alpha\in\mathbb{F}$.
	
	\emph{Computational Multi-evaluation}: Let $S\subset [n]$ such that $|S| < n$. We have already established that $\{f_i\}_{i\in S}$ has the same distribution for all $a\in\{0,1\}^\ell$ and $\alpha\in\mathbb{F}$. Assume that $x_1,\,x_2,\,\ldots,\,x_m \neq a$ and $r_1,\,r_2,\,\ldots,\,r_m\in R$ are distinct. Then, we get
	\begin{align*}
		\mathsf{Eval}(f_i, x_h, &r_h)= \left(\textstyle{\sum}_{j=0}^{\ell-1}\left(v_{i,2j+x_{h,j}} + F_1(k_{i,2j+x_{h,j}},r_h)\right),\right. \left.\textstyle{\sum}_{j=0}^{\ell-1}\left(\alpha_{i,2j+x_{h,j}} + F_2(k_{i,2j+x_{h,j}},r_h)\right),r_h,\theta,k\right) \\
		=\ &\left(\textstyle{\sum}_{j=0}^{\ell-1}v_{i,2j+x_{h,j}} + F_1\left(\textstyle{\sum}_{j=0}^{\ell-1}k_{i,2j+x_{h,j}},r_h\right),\right. \left.\textstyle{\sum}_{j=0}^{\ell-1}\alpha_{i,2j+x_{h,j}} + F_2\left(\textstyle{\sum}_{j=0}^{\ell-1}k_{i,2j+x_{h,j}},r_h\right),r_h,\theta,k\right).
	\end{align*}
	Since $\{f_i\}_{i\in S}$ has the same distribution regardless of the choice of $a$ and $\alpha$, the same holds for $\left(\{f_i\}_{i\in S}, \{\mathsf{Eval}(f_i, x_h, r_h)\}_{i\in S, h\in [m]}\right)$. 
	
	We observe that, since $x_h\neq a$ for all $1\leq h\leq m$, for any fixed $h$, the set
	$$\{k_{i,j}: i\in S,\ 0\leq j\leq 2\ell-1\}\cup \{k\}\cup \{\textstyle{\sum}_{j=0}^{\ell-1}k_{i,2j+x_{h,j}}: i\not\in S\}$$
	is a set of random linearly independent elements in $\mathcal{K}$. Hence, any non-zero linear combination of $\{\textstyle{\sum}_{j=0}^{\ell-1}k_{i,2j+x_{h,j}}: i\not\in S\}$ is a uniformly random element in $\mathcal{K}$ that lies outside the span of $\{k\}\cup \{k_{i,j}: i\in S,\ 0\leq j\leq 2\ell-1\}$.
	
	Since, by assumption, any element in $\mathcal{K}$ has order at most $\gamma(\lambda)$, the span of $2\ell (n-1)+1$ elements has size at most $\gamma(\lambda)^{2\ell (n-1)+1}<\gamma(\lambda)^{2\ell n}$. By our assumption, $\frac{\gamma(\lambda)^{2\ell n}}{|\mathcal{K}^{(\lambda)}|}<1-\frac{1}{\lambda}$, so the advantage of an adversary in distinguishing the PRF $F$ from random when the key is selected from outside the span of $\{k\}\cup \{k_{i,j}: i\in T,\ 0\leq j\leq 2\ell-1\}$ is increased by a factor of at most $\lambda$, and hence this advantage is still negligible in $\lambda$.
	
	Hence, given $\{f_i\}_{i\in S}$, the set $\{F(\textstyle{\sum}_{j=0}^{\ell-1}k_{i,2j+x_{h,j}},r_h)\}_{i\not\in S, h\in [m]}$ cannot be distinguished from uniformly random except with negligible probability. It follows that for all $i\not\in S$ and $h\in [m]$, $\textstyle{\sum}_{j=0}^{\ell-1}v_{i,2j+x_{h,j}} + F_1\left(\textstyle{\sum}_{j=0}^{\ell-1}k_{i,2j+x_{h,j}},r_h\right)$ and $\textstyle{\sum}_{j=0}^{\ell-1}\alpha_{i,2j+x_{h,j}} + F_2\left(\textstyle{\sum}_{j=0}^{\ell-1}k_{i,2j+x_{h,j}},r_h\right)$ are indistinguishable from independent uniform random elements of $\mathbb{F}^{2\ell+\lambda}$ and $\mathbb{F}$ respectively, except with probability negligible in $\lambda$. $\qed$
\end{proof}

\subsection{$t$-out-of-$n$ Multi-evaluation DPF}
In this section, we introduce the idea of an $\mathbb{F}$-key-homomorphic PRF. By assuming the existence of such PRFs, we extend the $n$-out-of-$n$ scheme in the previous subsection to a $t$-out-of-$n$ computational multi-evaluation threshold DPF scheme.

\begin{definition}[$\mathbb{F}$-key-homomorphic PRF]
	\emph{Let $\mathbb{F}$ be a field, $\mathcal{K}$ and $\mathbb{L}$ be extension fields of $\mathbb{F}$, and $F: \mathcal{K}\times \mathcal{X} \rightarrow \mathbb{L}^m$ be an efficiently computable function. We say that $F$ is an $\mathbb{F}$-key-homomorphic PRF if the following properties hold:
		\begin{enumerate}
			\item $F$ is a secure PRF,
			\item $\forall k_1, k_2 \in \mathcal{K}, x \in \mathcal{X}: F_{k_1 + k_2}(x) = F_{k_1}(x) + F_{k_2}(x),$
			\item $\forall c \in \mathbb{F}, k \in \mathcal{K}, x \in \mathcal{X}: F_{ck}(x) = c\cdot F_{k}(x).$
	\end{enumerate}}
\end{definition}

\begin{remark}
	Note that if $\mathcal{K}$ and $\mathbb{L}$ are fields with the same prime subfield $\mathbb{F}_p$, then a key-homomorphic PRF $F:\mathcal{K}\times \mathcal{X} \rightarrow \mathbb{L}^m$ satisfying (2) is always $\mathbb{F}_p$-key-homomorphic. Furthermore, since $\mathcal{K}$ is a finite field, we know that $(\mathcal{K},+)\cong\mathbb{F}_{p'}^{l}$ for some prime $p'$. Then, it follows from \Cref{KH-PRF_remark} that:
	$$|\mathcal{K}/p\mathcal{K}|^{-1}=
	\begin{cases}
		1/p^l & \text{if }p=p', \\
		1 & \text{otherwise}.
	\end{cases}
	$$
	Since $F$ is a secure PRF, $|\mathcal{K}/p\mathcal{K}|^{-1}$ is a negligible function of $\lambda$, thus it must be the case that $p'=p$, i.e. $\Char({\mathcal{K})}$ must be equal to $\Char(\mathbb{L})$.
\end{remark}

We will use an $\mathbb{F}$-key-homomorphic PRF family to produce a computational multi-evaluation threshold DPF scheme for the class of point functions $\mathcal{P}\left(\{0,1\}^\ell,\mathbb{L}\right)$, where $\mathbb{L}=\mathbb{F}_q$ is the finite field with cardinality $q$. Assume $|\mathbb{F}|\geq n+1$, and fix an injection $\iota: \{0,1,\ldots,n\}\to\mathbb{F}$. We will use this injection to identity elements in $\{0,1,\ldots,n\}$ with elements of $\mathbb{F}$. Note that this injection need not be a homomorphism. Let $\mathcal{F}=\{F^{(\lambda)}: \mathcal{K}^{(\lambda)}\times R\to\mathbb{L}^{2\ell+\lambda+1}\}$
be a family of $\mathbb{F}$-key-homomorphic PRFs such that the advantage of any polynomial-time adversary in distinguishing $F^{(\lambda)}$ from random is negligible in $\lambda$. As above, we write $F^{(\lambda)}=(F^{(\lambda)}_1, F^{(\lambda)}_2)$.

Again, we will make the assumption that the order of any element in $\mathcal{K}^{(\lambda)}$ is bounded by some polynomial $\gamma(\lambda)$, from which it follows, without loss of generality, that $\frac{\gamma(\lambda)^{2\ell n}}{|\mathcal{K}^{(\lambda)}|}<1-\frac{1}{\lambda}$ holds for all $\lambda$. Our scheme is a collection of three algorithms, $(\mathsf{Gen}(a,\alpha,\lambda), \mathsf{Eval}(f_i, x, r), \mathsf{Rec}(\{s_i: i\in T\}))$, which are defined as:

$$\bm{\mathsf{Gen}(a,\alpha,\lambda)}$$
\begin{enumerate}[topsep=-3pt]
	\item Choose $2\ell$ random vectors $v_0, v_1, v_2, \ldots, v_{2\ell-1}$ from $\mathbb{L}^{2\ell+\lambda}$.
	\item\label{item2} Compute Shamir shares $v_{i,j}\in\mathbb{L}^{2\ell+\lambda}$ ($1\leq i\leq n$, $0\leq j\leq 2\ell-1$) for $v_i$. To be precise, for each $0\leq j\leq 2\ell -1$, randomly choose polynomials $r_{j,h}(X)\in \mathbb{L}[X]$ ($1\leq h\leq 2\ell+\lambda$), each of degree $\leq t-1$, such that $r_{j,h}(0)$ is equal to the $h$-th coordinate of $v_j$, and let the $h$-th coordinate of $v_{i,j}$ be $r_{j,h}(i)$.
	\item Let $a = a_1a_2\ldots a_\ell$ and compute $\theta=\sum_{j=0}^{\ell-1} v_{2j+a_j}$.
	\item Choose $2\ell$ random elements $\alpha_0, \alpha_1, \ldots, \alpha_{2\ell-1}\in\mathbb{L}$ subject to the condition $\alpha=\sum_{j=0}^{\ell-1}\alpha_{2j+a_j}$.
	\item Compute Shamir shares $\alpha_{i,j}\in\mathbb{L}$ ($1\leq i\leq n$, $0\leq j\leq 2\ell-1$) for $\alpha_j$, as in Step \ref{item2} above.
	\item Choose $2\ell n$ linearly independent keys $k_{i,j}$ ($1\leq i\leq n$, $0\leq j\leq 2\ell-1$).
	\item Choose random polynomials $p_{i,j}(X)\in\mathcal{K}[X]$ ($1\leq i\leq n$, $0\leq j\leq 2\ell -1$), each of degree $\leq t-1$ such that $p_{i,j}(0)=k_{i,j}$, and let $k_{i,j,l}=p_{i,j}(l)$ ($1\leq l\leq n$). Let
	$k^{(l)}=\{(i,j,k_{i,j,l}):1\leq i\leq n,\ 0\leq j\leq 2\ell-1\}.$
	\item Compute $k=\sum_{i=1}^n \sum_{j=0}^{\ell-1} k_{i,2j+a_j}$.
	\item Output $f_i=(i, v_{i,0}, \ldots, v_{i,2\ell-1}, \theta, \alpha_{i,0}, \ldots, \alpha_{i,2\ell-1}, k_{i,0}, \ldots, k_{i,2\ell-1}, k^{(i)}, k).$
\end{enumerate}

$$\bm{\mathsf{Eval}(f_i, x, r)}$$
\begin{enumerate}[topsep=-3pt]
	\item Parse $f_i$ as $(i, v_{i,0}, \ldots, v_{i,2\ell-1}, \theta, \alpha_{i,0}, \ldots, \alpha_{i,2\ell-1}, k_{i,0}, \ldots, k_{i,2\ell-1}, k^{(i)}, k).$
	\item Let $x = x_1x_2\ldots x_\ell$. Compute $s_{i,0}=\textstyle{\sum}_{j=0}^{\ell-1}\left(v_{i,2j+x_j} + \textstyle{\sum}_{l=1}^n F_1(k_{l,2j+x_j,i},r)\right).$
	\item Compute $s_{i,1}=\textstyle{\sum}_{j=0}^{\ell-1}\left(\alpha_{i,2j+x_j} +  \textstyle{\sum}_{l=1}^n F_2(k_{l,2j+x_j,i},r)\right).$
	\item Output $s_i=(i, s_{i,0}, s_{i,1}, r, \theta, k)$
\end{enumerate}

$$\bm{\mathsf{Rec}(\{s_i: i\in T\})}$$
\begin{enumerate}[topsep=-3pt]
	\item Parse $s_i$ as $(i, s_{i,0}, s_{i,1}, r, \theta, k)$.
	\item Compute $S_{0,1}(X),\,\ldots,\,S_{0,2\ell+\lambda}(X)$ and $S_1(X)$, polynomials of degree $\leq t-1$ such that $S_{0,h}(i)$ is equal to the $h$-th coordinate of $s_{i,0}$ and $S_1(i)=s_{i,1}$ for all $i\in T$.
	\item If $\left(S_{0,1}(0),\ldots,S_{0,2\ell+\lambda}(0)\right)$ equals $\theta+F_1(k,r)$, output $S_1(0)-F_2(k,r)$ else output $0$.
\end{enumerate}

\begin{remark}
	Each party has a share size of
	$$\log n+(4\ell^2+2\lambda\ell+4\ell+\lambda)\log|\mathbb{L}|+(2\ell n+2\ell+1)\log|\mathcal{K}|$$
	(assuming we fix an ordering for the elements in $k^{(l)}$, and replace $(i,j,k_{i,j,l})$ by $k_{i,j,l}$)
	and the output of $\mathsf{Eval}$ has size
	$$\log n+(4\ell+2\lambda+1)\log|\mathbb{L}|+\log|\mathcal{K}|+\log|R|.$$
\end{remark}

\begin{theorem}
	\label{t_out_of_n_DPF}
	The above scheme is an $t$-out-of-$n$ computational multi-evaluation DPF scheme for sharing the class of point functions $\mathcal{P}\left(\{0,1\}^\ell,\mathbb{L}\right)$.
\end{theorem}

Before proving the above theorem, we prove two useful lemmas:

\begin{lemma}
	\label{lagrange_interpolation_lemma}
	Let $x_0, x_1,\ldots, x_t\in\mathbb{F}$ be distinct, $\mathcal{K}$ be an extension field of $\mathbb{F}$ and let $p(X)\in\mathcal{K}[X]$ be a polynomial of degree $\leq t-1$. Then there exists $c_1,c_2,\ldots,c_t\in\mathbb{F}$ such that
	$$p(x_0)=c_1p(x_1)+\cdots+c_tp(x_t).$$ 
\end{lemma}
\begin{proof}
	By Lagrange interpolation,
	$$p(X)=p(x_1)\cdot\frac{\prod_{i\neq 1}(X-x_i)}{\prod_{i\neq 1}(x_1-x_i)}+\cdots+p(x_t)\cdot\frac{\prod_{i\neq t}(X-x_i)}{\prod_{i\neq t}(x_t-x_i)},$$
	so
	$$p(x_0)=p(x_1)\cdot\frac{\prod_{i\neq 1}(x_0-x_i)}{\prod_{i\neq 1}(x_1-x_i)}+\cdots+p(x_t)\cdot\frac{\prod_{i\neq t}(x_0-x_i)}{\prod_{i\neq t}(x_t-x_i)}.$$
	It is clear that $c_j=\frac{\prod_{i\neq j}(x_0-x_i)}{\prod_{i\neq j}(x_j-x_i)}$ lies in the subfield $\mathbb{F}$ since $x_0, x_1,\ldots, x_t\in\mathbb{F}$. $\qed$
\end{proof}

\begin{lemma}
	\label{lagrange_interpolation_corollary}
	Let $F:\mathcal{K}\times X\to\mathbb{L}$ be an $\mathbb{F}$-key-homomorphic PRF, $x_0, x_1,\ldots, x_t\in\mathbb{F}$ be distinct, and $p(X)\in\mathcal{K}[X]$ be a polynomial of degree $\leq t-1$. Then
	\begin{enumerate}[label=(\alph*)]
		\item $F(p(x_0),r)$ is an $\mathbb{F}$-linear combination of $F(p(x_i),r)$ $(1\leq i\leq t)$,
		\item there exists a polynomial $\delta(X)\in\mathbb{L}[X]$ of degree $\leq t-1$ such that $\delta(x_i)=F(p(x_i),r)$ for all $0\leq i\leq t$.
	\end{enumerate}
\end{lemma}
\begin{proof}
	\begin{enumerate}[label=(\alph*)]
		\item Since $F$ is $\mathbb{F}$-key-homomorphic,
		$$F(p(x_0),r)=c_1\cdot F(p(x_1),r)+\cdots+c_t\cdot F(p(x_t),r),$$
		where $c_1,\ldots c_t\in\mathbb{F}$ are as in \Cref{lagrange_interpolation_lemma}.
		\item Let $\delta(X)$ be the polynomial
		$$\delta(X)=F(p(x_1),r)\cdot\frac{\prod_{i\neq 1}(X-x_i)}{\prod_{i\neq 1}(x_1-x_i)}+\cdots+F(p(x_t),r)\cdot\frac{\prod_{i\neq t}(X-x_i)}{\prod_{i\neq t}(x_t-x_i)}.$$
		It is clear that $\delta(x_i)=F(p(x_i),r)$ for $1\leq i\leq t$. And, by the proof of (a),
		\begin{align*}
			F(p(x_0),r)&=F(p(x_1),r)\cdot\frac{\prod_{i\neq 1}(x_0-x_i)}{\prod_{i\neq 1}(x_1-x_i)}+\cdots+F(p(x_t),r)\cdot\frac{\prod_{i\neq t}(x_0-x_i)}{\prod_{i\neq t}(x_t-x_i)} \\
			&=\delta(x_0). \qquad \qquad \qquad \qquad \qquad \qquad \qquad \qquad \qquad \qquad \qquad \quad \qed 
		\end{align*}
	\end{enumerate}
\end{proof}

\begin{proof}[of \Cref{t_out_of_n_DPF}]
	\emph{Computational Correctness}: Let $T$ be a subset of $[n]$ of size $t$. Without loss of generality, let us assume $T=[t]$. We start by proving that
	$$\mathsf{Rec}\left(\{\mathsf{Eval}(f_i, a, r)\}_{i\in T}\right) = \alpha.$$
	Note that for all $1\leq i\leq t$ and $1\leq h\leq 2\ell+\lambda$, 
	\begin{align*}
		S_{0,h}(i)=s_{i,0}[h]&=\textstyle{\sum}_{j=0}^{\ell-1}\left(v_{i,2j+a_j}[h]+\textstyle{\sum}_{l=1}^n F_1(k_{l,2j+a_j,i},r)[h]\right) \\
		&=\textstyle{\sum}_{j=0}^{\ell-1}r_{2j+a_j,h}(i)+\textstyle{\sum}_{j=0}^{\ell-1}\textstyle{\sum}_{l=1}^n F_1(p_{l,2j+a_j}(i),r)[h] \\
		&=\textstyle{\sum}_{j=0}^{\ell-1}r_{2j+a_j,h}(i)+F_1(\textstyle{\sum}_{j=0}^{\ell-1}\textstyle{\sum}_{l=1}^n p_{l,2j+a_j}(i),r)[h].
	\end{align*}
	Let $p(X)=\sum_{j=0}^{\ell-1}\sum_{l=1}^n p_{l,2j+a_j}(X)$, a polynomial of degree $\leq t-1$. By \Cref{lagrange_interpolation_corollary}(b), there exists a polynomial $\delta_h(X)\in\mathbb{L}[X]$ of degree $\leq t-1$, such that $\delta_h(i)=F_1(p(i),r)[h]$ for all $0\leq i\leq t$. Since $S_{0,h}(X)$ agrees with $\sum_{j=0}^{\ell-1}r_{2j+a_j,h}(X)+\delta_h(X)$ at the $t$ points $X=1,2,\ldots, t$, and both of them are polynomials of degree $\leq t-1$, they must be identical, i.e. $S_{0,h}(X)=\textstyle{\sum}_{j=0}^{\ell-1}r_{2j+a_j,h}(X)+\delta_h(X).$ Therefore,
	\begin{align*}
		S_{0,h}(0)&=\textstyle{\sum}_{j=0}^{\ell-1}r_{2j+a_j,h}(0)+\delta_h(0) \\
		&=\textstyle{\sum}_{j=0}^{\ell-1}v_{2j+a_j}[h]+F_1(p(0),r)[h] \\
		&=\textstyle{\sum}_{j=0}^{\ell-1}v_{2j+a_j}[h]+F_1(\textstyle{\sum}_{j=0}^{\ell-1}\textstyle{\sum}_{l=1}^n p_{l,2j+a_j}(0),r)[h] \\
		&=\theta[h]+F_1(\textstyle{\sum}_{j=0}^{\ell-1}\textstyle{\sum}_{l=1}^n k_{l,2j+a_j},r)[h] \\
		&=\theta[h]+F_1(k,r)[h],
	\end{align*}
	i.e., $\left(S_{0,1}(0),\ldots,S_{0,2\ell+\lambda}(0)\right)=\theta+F_1(k,r)$. The output of $\mathsf{Rec}$ is thus $S_1(0)-F_2(k,r)$, which, by a similar argument as above, is equal to:
	$$(\textstyle{\sum}_{j=0}^{\ell-1}\alpha_{2j+a_j}+F_2(k,r))-F_2(k,r)=\textstyle{\sum}_{j=0}^{\ell-1}\alpha_{2j+a_j}=\alpha.$$
	
	Next, we prove that evaluation at $x\neq a$ is correct except with probability negligible in $\lambda$. Let $u_j=v_j+F_1(k_j,r)$ for $j=0,\ldots, 2\ell-1$, where $k_j=\sum_{i=1}^n k_{i, j}$. Again, by a similar argument as above,
	\begin{align*}
		\textstyle{\sum}_{j=0}^{\ell-1}u_{2j+a_j}&=\textstyle{\sum}_{j=0}^{\ell-1}\left(v_{2j+a_j}+F_1(\textstyle{\sum}_{l=1}^n k_{l,2j+a_j},r)\right)=\theta+F_1(k,r),
	\end{align*}
	and evaluation at $x\neq a$ gives:
	\[
	\left(S_{0,1}(0),\ldots,S_{0,2\ell+\lambda}(0)\right) =\textstyle{\sum}_{j=0}^{\ell-1}\left(v_{2j+x_j}+F_1(\textstyle{\sum}_{l=1}^n k_{l,2j+x_j},r)\right) =\textstyle{\sum}_{j=0}^{\ell-1}u_{2j+x_j}.
	\]
	The result now follows by following the proof for computational correctness in \Cref{n_out_of_n_DPF}.
	
	\emph{Perfect Secrecy}: $\mathsf{Gen}(a,\alpha,\lambda)$ outputs $(f_1,f_2,\ldots,f_n)$, where
	$$f_i=(i, v_{i,0}, \ldots, v_{i,2\ell-1}, \theta, \alpha_{i,0}, \ldots, \alpha_{i,2\ell-1}, k_{i,0}, \ldots, k_{i,2\ell-1}, k^{(i)}, k)$$
	and $k^{(i)}=\{(i',j',k_{i',j',i}):1\leq i'\leq n,\ 0\leq j'\leq 2\ell-1\}.$ \\
	Suppose we are given $f_i$ from $t-1$ parties, which we will assume, without loss of generality, to be the first $t-1$ parties. 
	Any $t-1$ Shamir shares of a $t$-out-of-$n$ threshold scheme are independently and uniformly distributed. Thus, $v_{i,j}$ ($1\leq i\leq t-1$, $0\leq j\leq 2\ell-1$) and $\theta$ are independently and uniformly distributed. The same holds for $\alpha_{i,j}$ ($1\leq i\leq t-1$, $0\leq j\leq 2\ell-1$). $k_{i,j}$ ($1\leq i\leq t-1$, $0\leq j\leq 2\ell-1$) and $k$ are $2\ell(t-1)+1$ linearly independent elements picked uniformly at random from $\mathcal{K}$, while $k_{i',j',i}$ ($1\leq i'\leq n$, $0\leq j'\leq 2\ell-1$, $1\leq i\leq t-1$) are independently and uniformly distributed. Thus, the distribution of $(f_1,\ldots,f_{t-1})$ does not depend on $a$ or $\alpha$.
	
	\emph{Computational Multi-evaluation}: Let $S\subset [n]$ such that $|S| < t$. Assume $x_1,\,x_2,\,\ldots,\,x_m \neq a$, and $r_1,\,r_2,\,\ldots,\,r_m\in R$ are distinct. We get
	\begin{align*}
		\mathsf{Eval}(f_i, x_h, r_h) =\ &\left(i, \textstyle{\sum}_{j=0}^{\ell-1}\left(v_{i,2j+x_{h,j}}+\textstyle{\sum}_{l=1}^n F_1(k_{l,2j+x_{h,j},i},r_h)\right),\right. \\
		&\qquad\left.\textstyle{\sum}_{j=0}^{\ell-1}\left(\alpha_{i,2j+x_{h,j}}+\textstyle{\sum}_{l=1}^n F_2(k_{l,2j+x_{h,j},i},r_h)\right),r_h,\theta,k\right) \\
		=\ &\left(i, \textstyle{\sum}_{j=0}^{\ell-1}v_{i,2j+x_{h,j}}+F_1(\textstyle{\sum}_{j=0}^{\ell-1}\textstyle{\sum}_{l=1}^n k_{l,2j+x_{h,j},i},r_h),\right. \\
		&\qquad\left.\textstyle{\sum}_{j=0}^{\ell-1}\alpha_{i,2j+x_{h,j}}+F_2(\textstyle{\sum}_{j=0}^{\ell-1}\textstyle{\sum}_{l=1}^n k_{l,2j+x_{h,j},i},r_h),r_h,\theta,k\right).
	\end{align*}
	
	Let $S\subseteq U\subseteq [n]$. We shall prove by induction on $|U|$ that the distribution of \\$\left(\{f_i\}_{i\in S}, \{\mathsf{Eval}(f_i, x_h, r_h)\}_{i\in U, h\in [m]}\right)$ is computationally indistinguishable regardless of the choice of $a$ and $\alpha$. The base case simply follows from perfect secrecy; since the distribution of $\{f_i\}_{i\in S}$ is independent of the choice of $a$ and $\alpha$, so is the distribution of $\left(\{f_i\}_{i\in S}, \{\mathsf{Eval}(f_i, x_h, r_h)\}_{i\in S, h\in [m]}\right)$.
	
	Suppose, for some $S\subseteq U'\subset [n]$, that $\left(\{f_i\}_{i\in S}, \{\mathsf{Eval}(f_i, x_h, r_h)\}_{i\in U', h\in [m]}\right)$ is computationally indistinguishable regardless of the choice of $a$ and $\alpha$. Let $u\not\in U'$, and let $U^*=U'\cup\{u\}$. We consider the following three cases.\\[1.5mm]
	\textbf{Case 1:} $|U'|<t-1$, i.e. $|U^*|<t$. It follows from perfect secrecy that the distribution of\\ $\left(\{f_i\}_{i\in S}, \{\mathsf{Eval}(f_i, x_h, r_h)\}_{i\in U^*, h\in [m]}\right)$ is independent of $a$ and $\alpha$.\\[1.5mm]
	\textbf{Case 2:} $|U'|=t-1$, i.e. $|U^*|=t$. Assume we are given:
	$$\left(\{f_i\}_{i\in U'}, \{\mathsf{Eval}(f_i, x_h, r_h)\}_{i\in U', h\in [m]}\right),$$
	whose distribution is independent of $a$ and $\alpha$ by perfect secrecy.
	
	Fix some $h\in [m]$. Note that $\textstyle{\sum}_{j=0}^{\ell-1}\textstyle{\sum}_{l=1}^n k_{l,2j+x_{h,j},i}$ ($i\in U^*$) are Shamir shares of $\textstyle{\sum}_{j=0}^{\ell-1}\textstyle{\sum}_{l=1}^n k_{l,2j+x_{h,j}}$, which is randomly and uniformly distributed as an element of $\mathcal{K}$ outside the span of $\{k_{i,j}: i\in U',\ 0\leq j\leq 2\ell-1\}\cup \{k\}$. It follows that the advantage of an adversary in distinguishing $F(\textstyle{\sum}_{j=0}^{\ell-1}\textstyle{\sum}_{l=1}^n k_{l,2j+x_{h,j}}, r)$ from random is negligible in $\lambda$, thus the same holds for $F(\textstyle{\sum}_{j=0}^{\ell-1}\textstyle{\sum}_{l=1}^n k_{l,2j+x_{h,j},u},r)$ (which, by \Cref{lagrange_interpolation_corollary}(a), is an $\mathbb{F}$-linear combination of $F(\textstyle{\sum}_{j=0}^{\ell-1}\textstyle{\sum}_{l=1}^n k_{l,2j+x_{h,j}}, r)$ and $F(\textstyle{\sum}_{j=0}^{\ell-1}\textstyle{\sum}_{l=1}^n k_{l,2j+x_{h,j},i}, r)$ for $i\in U'$).
	
	Thus, even with knowledge of $\left(\{f_i\}_{i\in U'}, \{\mathsf{Eval}(f_i, x_h, r_h)\}_{i\in U', h\in [m]}\right)$, the distributions\\ $\textstyle{\sum}_{j=0}^{\ell-1}v_{u,2j+x_{h,j}}+F_1(\textstyle{\sum}_{j=0}^{\ell-1}\textstyle{\sum}_{l=1}^n k_{l,2j+x_{h,j},u},r_h)$ and $\textstyle{\sum}_{j=0}^{\ell-1}\alpha_{u,2j+x_{h,j}}+F_2(\textstyle{\sum}_{j=0}^{\ell-1}\textstyle{\sum}_{l=1}^n k_{l,2j+x_{h,j},u},r_h)$ are indistinguishable from uniformly random, except with probability negligible in $\lambda$.\\[1.5mm]
	\textbf{Case 3:} $|U'|\geq t$, i.e. $|U^*|>t$. Assume we are given:
	$$\left(\{f_i\}_{i\in S}, \{\mathsf{Eval}(f_i, x_h, r_h)\}_{i\in U', h\in [m]}\right),$$
	whose distribution is computationally independent of $a$ and $\alpha$ by the induction hypothesis. Since $\textstyle{\sum}_{j=0}^{\ell-1}v_{i,2j+x_{h,j}}+F_1(\textstyle{\sum}_{j=0}^{\ell-1}\textstyle{\sum}_{l=1}^n k_{l,2j+x_{h,j},i},r_h)$ are Shamir shares of $\textstyle{\sum}_{j=0}^{\ell-1}v_{2j+x_{h,j}}+F_1(\textstyle{\sum}_{j=0}^{\ell-1}\textstyle{\sum}_{l=1}^n k_{l,2j+x_{h,j}},r_h)$, it follows from \Cref{lagrange_interpolation_lemma} that for any $u_1,\ldots, u_t\in U'$, there exists $c_1, \ldots, c_t\in\mathbb{F}$ such that:
	\begin{align*}
		&\textstyle{\sum}_{j=0}^{\ell-1}v_{u,2j+x_{h,j}}+F_1(\textstyle{\sum}_{j=0}^{\ell-1}\textstyle{\sum}_{l=1}^n k_{l,2j+x_{h,j},u},r_h) \\
		=\ &c_1\cdot\left(\textstyle{\sum}_{j=0}^{\ell-1}v_{u_1,2j+x_{h,j}}+F_1(\textstyle{\sum}_{j=0}^{\ell-1}\textstyle{\sum}_{l=1}^n k_{l,2j+x_{h,j},u_1},r_h)\right) \\
		&\qquad+\cdots+c_t\cdot\left(\textstyle{\sum}_{j=0}^{\ell-1}v_{u_t,2j+x_{h,j}}+F_1(\textstyle{\sum}_{j=0}^{\ell-1}\textstyle{\sum}_{l=1}^n k_{l,2j+x_{h,j},u_t},r_h)\right).
	\end{align*}
	A similar argument shows that $\textstyle{\sum}_{j=0}^{\ell-1}\alpha_{u,2j+x_{h,j}}+F_2(\textstyle{\sum}_{j=0}^{\ell-1}\textstyle{\sum}_{l=1}^n k_{l,2j+x_{h,j},u},r_h)$ (i.e., $\mathsf{Eval}(f_u, x_h, r_h)$) is determined by $\left(\{f_i\}_{i\in S}, \{\mathsf{Eval}(f_i, x_h, r_h)\}_{i\in U', h\in [m]}\right)$. $\qed$
\end{proof}

\section{Function-private CDS}\label{sec4}
In this section, we introduce the concept of a \textit{function-private} CDS (FPCDS) scheme. We start with an informal description. Suppose we have $k$ parties $P_1, \ldots, P_k$. Let $h:\mathcal{C}\to\{0,1\}$, where $\mathcal{C} = \mathcal{C}_1 \times \mathcal{C}_2 \times \ldots \times \mathcal{C}_k$, be a Boolean condition function that lies in a family of Boolean condition functions $\mathcal{H}$. An FPCDS scheme for $h \in \mathcal{H}$ with secret domain $\mathcal{S}$ consists of a dealer $\mathcal{D}$, $k$ parties $P_1, \ldots, P_k$ and a third party, called Carol, which possesses an algorithm, $\mathsf{Carol}$.

The dealer runs a randomized $\mathsf{Gen}$ algorithm with the inputs $h$ and $s$, and obtains $w_1, w_2, \ldots, w_k$. For $1 \leq j \leq k$, $\mathcal{D}$ sends $P_j$ the output portion $w_j$. Then, each player $P_j$ chooses some $c_j \in \mathcal{C}_j$ as their portion of the condition. Next, $P_j$ sends some message $m_j$ to Carol where $m_j = \mathsf{P}_j(c_j, w_j)$ is the output of some party-specific algorithm run on their part $c_j$ of the input $c$ and their part $w_j$ of the $\mathsf{Gen}$ output. Upon receiving the messages $m_j$, Carol runs $\mathsf{Carol}(m_1, m_2, \ldots, m_k)$; it accepts (and outputs the secret) or rejects based on the output of its $\mathsf{Carol}$ algorithm. This scheme satisfies certain correctness, privacy and secrecy properties, as detailed in the following definition:

\begin{definition}\label{FPCDS}
\emph{Let $\mathcal{H}$ be a family of Boolean condition functions, where each $h\in\mathcal{H}$ is a function from $\mathcal{C}=\mathcal{C}_1\times\cdots\times\mathcal{C}_k$ to $\{0,1\}$. A function-private CDS scheme for the family $\mathcal{H}$ of condition functions and secret domain $\mathcal{S}$ is defined as a collection of algorithms $(\mathsf{Gen}, \mathsf{P}_1,\ldots, \mathsf{P}_k, \mathsf{Carol})$ such that:}
	\begin{itemize} 
		\item \emph{$\mathsf{Gen}$ is a randomized algorithm that takes two inputs, $h\in\mathcal{H}$ and $s\in\mathcal{S}$, and generates $k$ shares $\{w_i\}_{i=1}^k$.
			\item For each $i\in[k]$, $\mathsf{P}_i$ is a deterministic algorithm that takes two inputs, a share $w_i$ and $c_i\in C_i$, and outputs a message $m_i$.
			\item $\mathsf{Carol}$ is a deterministic algorithm that takes $m_1,\ldots, m_k$ as inputs and outputs either an element of $\mathcal{S}$ or $\bot$.}
	\end{itemize}
\emph{These algorithms satisfies the following four conditions:}
\begin{itemize}
	\item Perfect Correctness: \emph{For every $h\in\mathcal{H}$, $s \in \mathcal{S}$, and $c = (c_1, c_2, \ldots, c_k)\in \mathcal{C}$, when $\mathsf{Gen}(h, s) = (w_1, w_2, \ldots, w_k)$ and $m_j = \mathsf{P}_j(c_j, w_j)$,
		$$\mathsf{Carol}(m_1, m_2, \ldots, m_k) =
		\begin{cases}
			s & \text{if }h(c)=1, \\
			\bot & \text{if }h(c)=0.
		\end{cases}
		$$
	}
	\item Perfect Secrecy: \emph{Fix $h \in \mathcal{H}$. For every $c = (c_1, c_2, \ldots, c_k) \in \mathcal{C}$, and any pair of secrets $s, s' \in \mathcal{S}$, let $\mathsf{Gen}(h, s) = (w_1^{(s)}, w_2^{(s)}, \ldots, w_k^{(s)})$, $\mathsf{Gen}(h, s') = (w_1^{(s')}, w_2^{(s')}, \ldots, w_k^{(s')})$, $m_j^{(s)} = \mathsf{P}_j(c_j, w_j^{(s)})$ and $m_j^{(s')} = \mathsf{P}_j(c_j, w_j^{(s')}).$ If $h(c) = 0$, then $$(m_1^{(s)}, m_2^{(s)}, \ldots, m_k^{(s)}) \text{ and } (m_1^{(s')}, m_2^{(s')}, \ldots, m_k^{(s')})$$ are perfectly indistinguishable.
	}
	\item Perfect Input Privacy: \emph{Let $h\in\mathcal{H}$ and $s \in \mathcal{S}$. Let $c = (c_1, c_2, \ldots, c_k) \in \mathcal{C}$ and $c' = (c_1', c_2', \ldots, c_k') \in \mathcal{C}$. Let $\mathsf{Gen}(h, s) = (w_1, w_2, \ldots, w_k)$, $m_j=\mathsf{P}_j(c_j,w_j)$ and $m_j'=\mathsf{P}_j(c_j',w_j)$. If $h(c)=h(c')$, then
	$$(m_1, \ldots, m_k)\text{ and } (m_1', \ldots, m_k')$$ are perfectly indistinguishable.
	}
	\item Perfect Function Privacy: \emph{Fix $c = (c_1, c_2, \ldots, c_k) \in \mathcal{C}$, $s \in \mathcal{S}$ and $i\in[k]$. Let $h, h' \in \mathcal{H}$ such that $h(c) = h'(c)$ and such that for every $c_i'\in \mathcal{C}_i$,
	$$\{h(c_1',c_2',\ldots,c_k'): c_j'\in \mathcal{C}_j\text{ for }j\neq i\}=\{h'(c_1',c_2',\ldots,c_k'): c_j'\in \mathcal{C}_j\text{ for }j\neq i\}.$$
	Let $\mathsf{Gen}(h, s) = (w_1^{(h)}, w_2^{(h)}, \ldots, w_k^{(h)})$ and $\mathsf{Gen}(h', s) = (w_1^{(h')}, w_2^{(h')}, \ldots, w_k^{(h')})$. For all $j = 1, \ldots, k$, let $m_j^{(h)}=\mathsf{P}_j(c_j,w_j^{(h)})$ and $m_j^{(h')}=\mathsf{P}_j(c_j,w_j^{(h')})$. Then,
	$$(c_i, s, w_i^{(h)}, m_1^{(h)}, \ldots, m_k^{(h)})\text{ and } (c_i, s, w_i^{(h')}, m_1^{(h')}, \ldots, m_k^{(h')})$$ are perfectly indistinguishable.
	}
\end{itemize}
\end{definition}

\subsection{A Simple FPCDS Scheme}

In this section, we present the first FPCDS scheme. Our scheme works with the family $\mathcal{H}=\{h_{(a,b)}: a,b\in\{0,1\}^n\}$ of Boolean condition functions, where $h_{(a,b)}: \{0,1\}^{2n} \to \{0,1\}$ is defined as:
$$h_{(a,b)}(\alpha,\beta)=
\begin{cases}
	1 & \text{if }(\alpha,\beta)=(a,b), \\
	0 & \text{otherwise},
\end{cases}
$$
and with secret domain $\mathcal{S} = \mathbb{G}$ where $\mathbb{G}$ is a finite Abelian group. Let $m_i[j]$ be the $j$th element of $m_i$ with index starting at 0.

	\begin{enumerate}
		\item $\mathcal{D}$ chooses a secret element $s \in \mathbb{G}$ and runs $\mathsf{Gen}(h_{(a,b)}, s)$. For this, $\mathcal{D}$ samples six random elements $t, r_1, r_2, u, v_1, v_2 \leftarrow \mathbb{G}$ such that $u$, $v_1$, $v_2$ are distinct. $\mathcal{D}$ sends $w_1=(a, s, t, r_1, u, v_1)$ to $P_1$ and $w_2=(b, s, t, r_2, u, v_2)$ to $P_2$.
		\item\label{step2} $P_1$ chooses $\alpha \in \{0,1\}^n$. If $\alpha = a$, then $P_1$ sends $m_1 = (u, s \oplus t)$ to Carol; otherwise, it sends $m_1 = (v_1, r_1)$. $P_2$ chooses $\beta$ in $\{0,1\}^n$. If $\beta = b$, then $P_2$ sends $m_2 = (u, t)$ to Carol; otherwise, it sends $m_2 = (v_2, r_2)$.
		\item Carol rejects if $m_1[0] \neq m_2[0]$; else it returns $g = \mathsf{Carol}(m_1, m_2) = m_1[1] \oplus m_2[1]$.
	\end{enumerate}

\begin{remark}
The communication complexity between the dealer and party $P_j$ is $|w_j| = n + 5|s|$, while the communication complexity between any party $P_j$ and Carol is $|m_j| = 2|s|$.
\end{remark}

\begin{theorem}
	The above scheme is a function-private CDS scheme.
\end{theorem}
\begin{proof}
	We prove that the scheme satisfies \Cref{FPCDS} for FPCDS. 
	
	\emph{Perfect Correctness:} Suppose the dealer chooses a secret $s \in \mathbb{G}$ and computes $\mathsf{Gen}(h_{(a,b)}, s) = (w_1,w_2)$, where $w_1=(a, s, t, r_1, u, v_1)$ and $w_2=(b, s, t, r_2, u, v_2)$. Suppose further that $P_1$ chooses $\alpha \in \{0,1\}^n$ and $P_2$ chooses $\beta \in \{0,1\}^n$. Let $m_1 = \mathsf{P}_1(\alpha, w_1)$ and $m_2 = \mathsf{P}_2(\beta, w_2)$. Suppose $h_{(a,b)}(\alpha, \beta) = 1$, i.e., $(\alpha, \beta) = (a,b)$. Then, $m_1 = (u, s \oplus t)$ and $m_2 = (u, t)$. Since $m_1[0] = u =m_2[0]$, Carol outputs: 
	$$g = m_1[1] \oplus m_2[1] = s\oplus t \oplus t = s.$$
	
	Else, suppose $h_{(a,b)}(\alpha, \beta) = 0$, i.e., $(\alpha, \beta) \neq (a,b)$. Then,
	$$(m_1[0], m_2[0]) = (u,v_2),\ (v_1, u)\text{ or }(v_1, v_2).$$
	By the choice of $u, v_1$, and $v_2$, all three elements are distinct. Hence, Carol rejects in all three cases.
	
	\emph{Perfect Secrecy:} Let $s, s' \in \mathbb{G}$, $\mathsf{Gen}(h_{(a,b)}, s) = (w_1,w_2)$ and $\mathsf{Gen}(h_{(a,b)}, s') = (w_1',w_2')$, where $w_1=(a, s, t, r_1, u, v_1)$, $w_2=(b, s, t, r_2, u, v_2)$, $w_1'=(a, s', t', r_1', u, v_1)$ and $w_2'=(b, s', t', r_2', u', v_2')$. For $\alpha,\beta \in \{0,1\}^n$, let $m_1 = \mathsf{P}_1(\alpha, w_1)$, $m_2 = \mathsf{P}_2(\beta, w_2)$, $m_1' = \mathsf{P}_1(\alpha, w_1')$ and $m_2' = \mathsf{P}_2(\beta, w_2')$. Suppose $h_{(a,b)}(\alpha,\beta)=0$, i.e., $(\alpha,\beta)\neq (a,b)$. We show that $(m_1, m_2)$ and $(m_1', m_2')$ are perfectly indistinguishable. For the first case, suppose $\alpha = a$ and $\beta \neq b$. Then $(m_1, m_2) = ((u, s \oplus t), (v_2, r_2))$ and $(m_1', m_2') = ((u', s' \oplus t'), (v_2', r_2')).$ In this case, since $u, u', t, t', r_2,r_2'$ are drawn uniformly, and $v_2, v_2'$ are drawn uniformly to be not equal to $u, u'$, respectively, $(m_1, m_2)$ and $(m_1', m_2'$) are both indistinguishable from $((\gamma, \delta),(\zeta, \eta))$, where $\gamma, \delta, \zeta, \eta \xleftarrow{\; \$ \;} \mathbb{G}$ such that $\gamma \neq \zeta$. The case where $\alpha\neq a$ and $\beta=b$ is analogous to this case.
	
	For the final case, suppose $\alpha\neq a$ and $\beta \neq b$. Then, it follows that $(m_1, m_2) = ((v_1, r_1), (v_2, r_2))$ and $(m_1', m_2') = ((v_1', r_1'), (v_2', r_2'))$. Again, $(m_1, m_2)$ and $(m_1', m_2'$) are both indistinguishable from $((\gamma, \delta), (\zeta, \eta))$, where $\gamma, \delta, \zeta, \eta \xleftarrow{\; \$ \;} \mathbb{G}$ such that $\gamma \neq \zeta$.
	
	\emph{Perfect Input Privacy:} Let $s \in \mathbb{G}$ and $\mathsf{Gen}(h_{(a,b)}, s) = (w_1, w_2)$ with $w_1=(a, s, t, r_1, u, v_1)$ and $w_2=(b, s, t, r_2, u, v_2)$. Suppose $\alpha, \alpha',\beta,\beta' \in \{0,1\}^n$ satisfy the  condition that $h_{(a,b)}(\alpha, \beta) = h_{(a,b)}(\alpha', \beta')$. Let $m_1 = \mathsf{P}_1(\alpha, w_1)$, $m_1' = \mathsf{P}_1(\alpha', w_1)$, $m_2 = \mathsf{P}_2(\beta, w_2)$ and $m_2' = \mathsf{P}_2(\beta', w_2)$. We shall show that $(m_1, m_2)$ and $(m_1', m_2')$ are indistinguishable. The only case where $h_{(a,b)}(\alpha,\beta)=1=h_{(a,b)}(\alpha',\beta')$ is when $(a,b)=(\alpha,\beta)=(a',b')$, so this case is trivial. We now consider $h_{(a,b)}(\alpha,\beta)=0=h_{(a,b)}(\alpha',\beta')$. For the first case, suppose $\alpha = \alpha' = a$. Then, $\beta \neq b \neq \beta'$, and hence $m_1 = m_1' = (u ,s \oplus t)$ with $m_2 = m_2' = (v_2, r_2)$. Thus, $(m_1, m_2)$ and $(m_1', m_2')$ are identical. 
	
	For the second case, suppose $\alpha = a$ and $\alpha' \neq a$ while $\beta \neq b$ and $\beta' = b$. Then, $(m_1, m_2) = ((u, s\oplus t), (v_2, r_2))$ while $(m_1', m_2') = ((v_1, r_1), (u, t))$. Since $u, t, v_1, v_2, r_1,$ and $r_2$ are all drawn uniformly, both $(m_1, m_2)$ and $(m_1', m_2')$ are indistinguishable from $((\gamma, \delta), (\zeta, \eta))$, where $\gamma, \delta, \zeta, \eta \xleftarrow{\; \$ \;} \mathbb{G}$ such that $\gamma \neq \zeta$. The rest of the cases are similar, and therefore the scheme satisfies perfect input privacy.
	
	\emph{Perfect Function Privacy:} Let $s \in \mathbb{G}$ and $\alpha, \beta \in \{0,1\}^n$. Let $h = h_{(a,b)}$ and $h' = h_{(c,d)} \in \mathcal{H}$ be point functions such that $h(\alpha, \beta) = h'(\alpha, \beta)$ and such that for every $c_1\in\{0,1\}^n$, $\{h(c_1,c_2): c_2\in C_2\}=\{h'(c_1,c_2): c_2\in C_2\}$.
	Let $\mathsf{Gen}(h, s) = (w_1,w_2)$, $\mathsf{Gen}(h', s) = (w_1',w_2')$ where $w_1=(a, s, t, u, r_1, v_1)$, $w_2=(b, s, t, u, r_2, v_2)$, $w_1'=(c, s, t', u', r_1', v_1')$ and $w_2'=(d, s, t', u', r_2', v_2')$. Let $m_1 = \mathsf{P}_1(h, w_1), m_2 = \mathsf{P}_2(h, w_2), m_1' = \mathsf{P}_1(h', w_1')$ and $m_2' = \mathsf{P}_2(h', w_2')$. We shall show that $(\alpha, s, w_1, m_1, m_2)$ and $(\alpha, s, w_1', m_1', m_2')$ are indistinguishable. If $h(\alpha, \beta) = h'(\alpha, \beta) = 1$, then the functions are identical point functions. Thus, suppose $h(\alpha, \beta) = h'(\alpha, \beta) = 0$ so that $(a, b)\neq (\alpha, \beta) \neq (c, d)$. The condition that $\{h(c_1,c_2): c_2\in C_2\}=\{h'(c_1,c_2): c_2\in C_2\}$ for all $c_1\in\{0,1\}^n$ means that $a=c$.
	
	First, suppose $a\neq \alpha \neq c$ and $b \neq \beta \neq d$. Then, $(m_1, m_2) = ((v_1, r_1), (v_2, r_2))$ and $(m_1', m_2') = ((v_1', r_1'), (v_2', r_2'))$. Since $u, v_1, v_2, r_1, r_2, u', v_1', v_2', r_1', r_2'$ are all drawn uniformly from $\mathbb{G}$ subject to the restrictions that $u$, $v_1$, $v_2$ and $u'$, $v_1$, $v_2'$ are distinct, respectively, it follows that both $(\alpha, s, (a, s, t, u, r_1, v_1), \allowbreak (v_1,r_1), (v_2,r_2))$ and $(\alpha, s, (a, s, t', u', r_1', v_1'), (v_1',r_1'), (v_2',r_2'))$ are indistinguishable from $(\alpha, s,\linebreak[4] (a, s, \theta, \lambda, \delta, \gamma), (\gamma,\delta), (\zeta, \eta))$, where $\gamma, \delta, \zeta, \eta, \theta, \lambda \xleftarrow{\; \$ \;} \mathbb{G}$ such that $\gamma$, $\zeta$ and $\lambda$ are distinct.
	
	Next, suppose that $a \neq \alpha \neq c$ and $b =\beta \neq d$. This gives $(m_1, m_2) = ((v_1, r_1), (u, t))$ and $(m_1', m_2') = ((v_1', r_1'), (v_2', r_2'))$. This case is similar to the one discussed above.
	
	Finally, suppose $a=\alpha=c$ and $b\neq\beta \neq d$. In this case, $(m_1, m_2) = ((u, s\oplus t), (v_2, r_2))$ and $(m_1', m_2') = ((u', s\oplus t'), (v_2', r_2'))$. Then, both $(\alpha, s, (a, s, t, u, r_1, v_1),\allowbreak (u, s\oplus t), (v_2, r_2))$ and $(\alpha, s,\allowbreak (a, s, t', u', r_1', v_1'), (u', s\oplus t'), (v_2', r_2'))$ are indistinguishable from $(\alpha, s, (a, s, \delta, \gamma, \theta, \lambda), (\gamma, s\oplus\delta), (\zeta, \eta))$, where $\gamma, \delta, \zeta, \eta, \theta, \lambda \xleftarrow{\; \$ \;} \mathbb{G}$ such that $\gamma$, $\zeta$ and $\lambda$ are distinct. $\qed$
\end{proof}

\begin{remark}
In the above scheme, secrecy, input privacy and function privacy may no longer hold if $P_1$ and $P_2$ repeat Step \ref{step2} of the protocol with the same shares $w_1$ and $w_2$. However, we can remove this limitation if we allow $P_1$ and $P_2$ to ``refresh'' their shares.

To do so, we fix a PRF $F:\mathcal{K}\times\mathbb{G}\to\mathbb{G}$, a PRP $P:\mathcal{K}\times\mathbb{G}\to\mathbb{G}$ and a PRF $F':\mathcal{K}'\times\{0,1\}^\ell\to\mathcal{K}$, and assume that the dealer $\mathcal{D}$ deals a common key $k'\in\mathcal{K}'$ to both parties $P_1$ and $P_2$. Each of the parties also keeps a counter $c$ of the number of times Step \ref{step2} of the protocol has been performed. After each run of Step \ref{step2}, $P_i$ updates the share $w_i$ by doing the following:
\begin{enumerate}
\item Compute $k_{1,i}=F'(k', c||1||i)$ and $k_j=F'(k', c||j)$ for $j=2,3$.
\item Replace $r_i$, $t$, $v_i$ and $u$ by $F(k_{1,i},r_i)$, $F(k_2, t)$, $P(k_3, v_i)$ and $P(k_3, u)$, respectively.
\item Increment $c$.
\end{enumerate}

Note that the keys $k_{1,1}$, $k_{1,2}$, $k_2$ and $k_3$ are unique to each run of the protocol. This is necessary as otherwise, if $|\mathbb{G}|$ is small, the elements $r_i$, $t$, $v_i$ and $u$ could end up in a cycle, repeating after a limited number of runs.
\end{remark}

\subsection{From FPCDS to FSS}
Let there be an FPCDS scheme for $k$ parties $P_1, \ldots, P_k$ for the family of Boolean condition functions $\mathcal{H}$, condition domain $\mathcal{C} = \mathcal{C}_1 \times \ldots \times \mathcal{C}_k$, and secret domain $\mathcal{S}$. Let $\mathcal{D}$ be the dealer for the FPCDS scheme and $\mathsf{Gen}$ be the randomized algorithm from the FPCDS scheme. We demonstrate that there exists a $k$-out-of-$k$ FSS scheme for the family of functions $\mathcal{H}$. Let $\mathcal{D}'$ be the dealer and $P_1', \ldots, P_k'$ be the $k$ parties for the FSS scheme. Our scheme is defined as a collection of three algorithms $(\mathsf{KeyGen}(h), \mathsf{Eval}(k_j), \mathsf{Rec}(m_1, \ldots, m_k))$, which are defined as:

		$$\bm{\mathsf{KeyGen}(h)}$$
		\begin{enumerate}[topsep=-3pt]
			\item For the chosen condition function $h \in \mathcal{H}$, $\mathcal{D}'$ samples a random secret $s \in \mathcal{S}$ and runs $\mathsf{Gen}(h, s)$ to generate the tuple $(w_1, w_2, \ldots, w_k)$.
			\item $\mathcal{D}'$ distributes to each $P_j'$ the key $k_j = w_j.$
		\end{enumerate} 

		$$\bm{\mathsf{Eval}(k_j)}$$
		\begin{enumerate}[topsep=-3pt]
			\item Each party $P_j'$, given $k_j = w_j$, chooses an input $c_j \in \mathcal{C}_j$.
			\item Each $P_j'$ runs the party algorithm $m_j = P_j(c_j, w_j)$ from the FPCDS scheme. 
		\end{enumerate}

		$$\bm{\mathsf{Rec}(m_1, \ldots, m_k)}$$
		\begin{enumerate}[topsep=-3pt]
			\item The $k$ parties publish the messages $m_1, m_2, \ldots, m_k$.
			\item The $k$ parties simulate Carol and compute $\mathsf{Carol}(m_1, m_2, \ldots, m_k$). If the $\mathsf{Carol}$ algorithm rejects, then the parties output 0; otherwise, if the $\mathsf{Carol}$ algorithm outputs the correct secret, $s$, the parties output 1.			
		\end{enumerate}

\subsubsection{Correctness and Privacy.} Our scheme satisfies the correctness and function privacy requirements for an FSS scheme.

	\emph{Perfect Correctness:} For the chosen condition function $h \in \mathcal{H}$, suppose $\mathcal{D}'$ samples $s \in \mathcal{S}$ and runs $\mathsf{Gen}(h, s) = (w_1, w_2, \ldots, w_k)$. Then, each $P_j'$ chooses $c_j \in \mathcal{C}_j$ so that $m_j = \mathsf{P}_j(c_j, w_j)$. Now, if $h(c_1,c_2,\ldots,c_k) = 1$, then by perfect correctness of FPCDS, $\mathsf{Carol}(m_1, m_2, \ldots, m_k$) returns the correct secret $s$. Hence, the parties output $1$ during the reconstruction step. Similarly, when $h(y) = 0$, $\mathsf{Carol}(m_1, m_2, \ldots, m_k)$ always rejects, and hence the parties output $0$.

	\emph{Function Privacy:} Fix $c = (c_1, c_2, \ldots, c_k) \in \mathcal{C}$ and $i \in [k]$. Let $h, h' \in \mathcal{H}$ such that for every $c_i' \in \mathcal{C}_i$, the sets $\{h(c_1',c_2',\ldots,c_k'): c_j'\in \mathcal{C}_j\text{ for }j\neq i\}$ and $\{h'(c_1',c_2',\ldots,c_k'): c_j'\in \mathcal{C}_j\text{ for }j\neq i\}$ are equal. Select $s\in\mathcal{S}$, and let $\mathsf{Gen}(h, s) = (w_1, w_2, \ldots, w_k)$, $\mathsf{Gen}(h', s) = (w_1', w_2', \ldots, w_k')$. Suppose $m_i = \mathsf{P}_i(c_i, w_i)$ and $m_i' = \mathsf{P}_i(c_i, w_i')$. During evaluation and reconstruction, $P_i'$ can observe $(c_i, w_i, m_1, m_2, \ldots, m_k)$ and $(c_i, w_i', m_1', m_2', \ldots, m_k')$, which are perfectly indistinguishable by the function privacy property of the FPCDS scheme.
	
Note that the given procedure is an FSS scheme for $\mathcal{H}$ with the following two caveats.
\begin{enumerate}
\item The input $c_i$ must remain private to $P_i'$.
\item For each party $P_i'$, function privacy only holds for $h,h'\in\mathcal{H}$ such that:
$$\{h(c_1',c_2',\ldots,c_k'): c_j'\in \mathcal{C}_j\text{ for }j\neq i\}=\{h'(c_1',c_2',\ldots,c_k'): c_j'\in \mathcal{C}_j\text{ for }j\neq i\}$$
for every $c_i'\in\mathcal{C}_i$.
\end{enumerate}

\section{Optimal Threshold Function Secret Sharing of Polynomials}\label{sec5}
Let $\mathcal{D}$ denote the dealer, $p_1, \ldots, p_k$ denote the $k$ parties, and fix the threshold number $t \leq k$. Suppose we are working with polynomials over a field $\mathbb{F}$. Have the dealer $\mathcal{D}$ choose some polynomial $p(x) = a_nx^n + a_{n-1}x^{n-1} + \ldots + a_0$. We will construct a $t$-out-of-$k$ function secret sharing scheme for $p(x)$ by defining the algorithms \textbf{Gen} (generation [of keys]), \textbf{Eval} (evaluation [of shares]), and \textbf{Rec} (reconstruction [of the function evaluation]).

\subsubsection{Gen.}
The dealer $\mathcal{D}$ fixes random polynomials $q_n, q_{n-1}, \ldots, q_0$ of degree $t-1$ over $\mathbb{F}$ such that $q_n(0) = a_n, q_{n-1}(0) = a_{n-1}, \ldots, q_0(0) = a_0$. Then, the dealer distributes the key $K_i = \langle q_n(i), q_{n-1}(i), \ldots, q_0(i)\rangle$ to party $p_i$.

\subsubsection{Eval.}
Let $\hat{x} \in \mathbb{F}$ be the desired input to $p(x)$ so that all parties have $\hat{x}$. Have each party $p_i$ calculate the vector $\textbf{x} = \langle \hat{x}^n, \hat{x}^{n-1}, \ldots, 1 \rangle$. Then, have each party $p_i$ calculate their share $s_i$ by taking the dot product of their key $K_i$ with the vector \textbf{x}, i.e. $s_i = K_i \cdot \textbf{x}$.

\subsubsection{Rec.}
When $t$ parties come together to reconstruct the output $p(\hat{x})$, they will have $t$ of the points on the polynomial $Q(y) = q_n(y)\hat{x}^n + q_{n-1}(y)\hat{x}^{n-1} + \ldots + q_0(y)$, and since $Q(y)$ is degree $t-1$ in $y$, the $t$ parties will be able to reconstruct the polynomial $Q(y)$ and therefore the value $Q(0) = p(\hat{x})$. 

\subsection{Correctness and Security}
We begin by analyzing the correctness of the scheme, followed by an evaluation of its security. 
\subsubsection{Correctness.}
Suppose $t$ parties (without loss of generality, let $p_1, \ldots, p_t$ be the parties) come together to reconstruct $p(\hat{x})$. Note that each party $p_i$ has evaluated 
$$K_i \cdot \textbf{x} = q_n(i)\hat{x}^n + q_{n-1}(i)\hat{x}^{n-1} + \ldots + q_0(i) = Q(i).$$
Thus, since all the polynomials $q_j$ are degree at most $t-1$, the $t$ parties cumulatively have $t$ points on the degree $t-1$ polynomial $Q(y)$ and therefore can reconstruct $Q(y)$. Now, note that 
$$Q(0) = q_n(0)\hat{x}^n + q_{n-1}(0)\hat{x}^{n-1} + \ldots + q_0(0) = a_n\hat{x}^n + a_{n-1}\hat{x}^{n-1} + \ldots + a_0 = p(\hat{x})$$
so that the evaluation of $Q(0)$ yields the correct output for $p(\hat{x})$ as desired. 

\subsubsection{Security.} Suppose at most $t-1$ parties come together to evaluate $p(\hat{x})$. Note that the reconstruction of $Q(y)$ is identical to Shamir reconstruction; determining $Q(y)$ (a degree $t-1$ polynomial) from $t-1$ points has the same security as Shamir secret sharing, so the scheme is secure.

\subsection{Optimality}
Let $\mathbb{F}_{q}$ be the finite field with $q$ elements for some prime power $q$, and let $\mathsf{Poly}_{\leq n}$ be the set of polynomials of degree $\leq n$ over $\mathbb{F}_q$. Note that sharing $\mathsf{Poly}_{\leq n}$ using the above FSS scheme results in a share size of $\log(q^{n+1})$ bits for each party. We will prove that this is optimal for an information-theoretical secure FSS scheme, at least for $n\leq q-1$.

It is clear that there are exactly $q^q$ functions from $\mathbb{F}_q$ to $\mathbb{F}_q$. Note that $|\mathsf{Poly}_{\leq q-1}|\leq q^q$ since a polynomial of degree $\leq q-1$ has $q$ coefficients, and each coefficient is an element of $\mathbb{F}_q$. Furthermore, since a polynomial of degree $\leq q-1$ has at most $q-1$ roots, any two polynomials in $\mathsf{Poly}_{\leq q-1}$ define distinct functions from $\mathbb{F}_q$ to $\mathbb{F}_q$. Thus, $\mathsf{Poly}_{\leq q-1}$ is precisely the set of all functions from $\mathbb{F}_q$ to $\mathbb{F}_q$.

We have established above that the polynomials in $\mathsf{Poly}_{\leq n}$ define distinct functions for $n\leq q-1$, thus there are $q^{n+1}$ distinct functions in $\mathsf{Poly}_{\leq n}$. Suppose there is a FSS scheme for $\mathsf{Poly}_{\leq n}$ where a certain party $P_j$ has a share size of $s<\log(q^{n+1})$ bits. Thus, the share of party $P_j$ is one of $2^s<q^{n+1}$ possibilities. Then given shares of any $t-1$ parties, by simply guessing the share of $P_j$, it is possible to narrow down the secret shared polynomial to one of $2^s<q^{n+1}$ possibilities, hence the scheme is not information-theoretic secure. This completes the proof of the claim.

\section{Future Work}\label{sec6}
FSS~\cite{Santis[94],Boyle[15]} aims to allow multiple parties to collaboratively evaluate a class $\mathcal{F}$ of functions $f: \{0,1\}^n \longrightarrow \mathbb{G}$, which is shared among $n \geq 2$ parties, by using distributed functions $f_i: \{0,1\}^\ell \longrightarrow \mathbb{G}$, where $\mathbb{G}$ is an abelian group such that (i) $\sum\limits_{i=1}^n f_i = f$, and (ii) any strict subset of $\{f_i\}_{i=1}^n$ hides $f$. Most of the known FSS schemes are $n$-out-of-$n$, i.e., they require all $n$ shares for function reconstruction. An FSS scheme is said to be threshold if for $t < n$, all subsets of cardinality at least $t$ can reconstruct $f$ from their shares, i.e., $\sum\limits_{i=1}^t f_i = f$. Similarly, if an FSS scheme allows sharing the function such that only certain selected subset(s) of arbitrary cardinality can reconstruct $f$, then it is called FSS for general access structures. The only FSS schemes for general access structures and threshold structures were developed by Koshiba~\cite{Kosh[18]} and Luo et al.~\cite{Luo[20]}, respectively. However, the scheme from~\cite{Kosh[18]} is restricted in that it only works for succinct functions (w.r.t. a fixed Fourier basis). 

Existing FSS solutions are not as flexible and versatile as the various secret sharing schemes~\cite{Beimel[11]}, which provide enhanced privacy~\cite{SehrawatVipin[21],Vipin[20],Stinson[87],Phillips[92],Blundo[96],Kishi[02],Deng[07],VipinThesis[19]}, refreshable shares~\cite{Yung[91],Amir[95],Frank[97],Zhou[05],Nikov[04]}, ability to share multiple secrets simultaneously~\cite{Yang[04],Blundo[94]}, verifiable shares and secrets~\cite{Tompa[89],McEliece[81],Chor[85],Cachin[02],Chaum[88],Ben[88],Feld[87],Oded[91],Tor[01],Rabin[89],Genn[98],Basu[19],Kate[10],Cascudo[17],Backes[13],Stadler[96]}, flexibility in sharing/reconstruction procedures and size of the shares~\cite{Shamir[79],Ito[87],Brick[89],Karch[93],Gusta[88],LiuStoc[18],Benny[20],Blundo[92],Capo[93],Csi[96],Csi[97],Dijk[95]}. Its extensive flavors have enabled secret sharing to have applications in a multitude of areas, including threshold cryptography~\cite{Yvo[89]}, (secure) multiparty computation~\cite{Ben[88],Chaum[88],Cramer[00],CramDam[15]}, attribute-based encryption~\cite{Goyal[06],Waterss[11]}, generalized oblivious transfer~\cite{Tassa[11],Shankar[08]}, perfectly secure message transmission~\cite{Danny[93]}, access control~\cite{Naor[06]}, e-voting~\cite{Berry[99],Yung[04]}, e-auctions~\cite{Mic[98],Peter[09]}, anonymous communications \cite{Sehrawat[17]}, and byzantine agreement~\cite{Abraham[08],Canetti[93],Feld[88],Katz[06],Patra[14]}. On the other hand, due to its limited and restrictive solutions, FSS has only found applications in secure computation~\cite{Boyle[19]}, private information retrieval~\cite{Boyle[16],Niv[14]} and private contact tracing~\cite{Ditti[20]}. Hence, there is a need to explore novel variants of FSS with enhanced flexibility and (function) privacy. Restricted functionality and computational security/privacy guarantees can be one approach to derive efficient FSS schemes with desired properties.

\bibliographystyle{plainurl}
\bibliography{CO2020}
\end{document}